\documentclass[letterpaper,11pt]{article}
\usepackage{fullpage}

\usepackage{amsmath, amscd, amssymb, amsthm}
\usepackage{hyperref}
\usepackage{cleveref}
\newtheorem{theorem}{Theorem}[section]

\newtheorem{lemma}[theorem]{Lemma}

\theoremstyle{definition}
\newtheorem{definition}[theorem]{Definition}

\newcommand{\EE}{\mathbb{E}}

\newcommand{\NN}{\mathbb{N}}

\newcommand{\RR}{\mathbb{R}}

\def\cM{{\mathcal M}}

 \DeclareMathOperator{\Lap}{Lap}
\DeclareMathOperator{\argmax}{argmax} 

\usepackage{algorithm, algorithmic}

\title{Distributed Private Heavy Hitters}
\author{Justin Hsu\thanks{Department of Computer and Information Sciences, University of Pennsylvania. Email: {\tt justhsu@cis.upenn.edu}} \and
 Sanjeev Khanna\thanks{Department of Computer and Information Sciences, University of Pennsylvania. Email: {\tt sanjeev@cis.upenn.edu}} \and
 Aaron Roth\thanks{Department of Computer and Information Sciences, University of Pennsylvania. Email: {\tt aaroth@cis.upenn.edu} Supported in part by NSF Awards CCF-1101389 and CNS-1065060.}}
\begin{document}
\maketitle

\begin{abstract}
In this paper, we give efficient algorithms and lower bounds for solving the \emph{heavy hitters} problem while preserving \emph{differential privacy} in the fully distributed \emph{local} model. In this model, there are $n$ parties, each of which possesses a single element from a universe of size $N$.  The heavy hitters problem is to find the identity of the most common element shared amongst the $n$ parties. In the local model, there is no trusted database administrator, and so the algorithm must interact with each of the $n$ parties separately, using a differentially private protocol. We give tight information-theoretic upper and lower bounds on the accuracy to which this problem can be solved in the local model (giving a separation between the local model and the more common centralized model of privacy), as well as computationally efficient algorithms even in the case where the data universe $N$ may be exponentially large.
\end{abstract}
%\begin{abstract}
%Our work concerns the heavy hitters problem, in a distributed and differentially
%private setting. We propose a compressed sensing inspired algorithm, and we
%analyze its performance. Next, we prove an information theoretic lower bound on
%error in this setting, showing that our algorithm achieves near optimal
%accuracy. Finally, we propose and analyze two efficient algorithms to solve this
%problem, with incomparable accuracy bounds.
%\end{abstract}

\section{Introduction}
Consider the problem of a website administrator who wishes to know what his most common traffic sources are. Each of $n$ visitors arrives with a single \emph{referring site}: the name of the last website that she visited, which is drawn from a vast universe $N$ of possible referring sites ($N$ here is the set of all websites on the internet). There is value in identifying the most popular referring site (the \emph{heavy hitter}): the site administrator may be able to better tailor the content of his webpage, or better focus his marketing resources. On the other hand, the identity of each individual's referring site might be embarrassing or otherwise revealing, and is therefore private information. We can therefore imagine a world in which this information must be treated ``privately.'' Moreover, in this situation, visitors are communicating directly with the servers of the websites that they visit: i.e. there is no third party who might be trusted to aggregate all of the referring website data and provide privacy preserving statistics to the website administrator. In this setting, how well can the website administrator estimate the heavy hitter while being able to provide formal privacy guarantees to his visitors?

This situation can more generally be modeled as the \emph{heavy hitters} problem under the constraint of \emph{differential privacy}. There are $n$ individuals $i \in [n]$ each of whom is associated with an element $v_i \in N$ of some large data universe $N$. The \emph{heavy hitter} is the most frequently occurring element $x\in N$ among the set $\{v_1,\ldots,v_n\}$, and we would like to be able to identify that element, or one that occurs almost as frequently as the heavy hitter. Moreover, we wish to solve this problem while preserving \emph{differential privacy} in the fully distributed (local) model. We define this formally in section \ref{sec:prelims}, but roughly speaking, an algorithm is differentially private if changes to the data of single individuals only result in small changes in the output distribution of the algorithm. Moreover, in the fully distributed setting, each individual (who can be viewed as a database of size 1) must interact with the algorithm independently of all of the other individuals, using a differentially private algorithm. This is in contrast to the more commonly studied centralized model, in which a trusted database administrator may have (exact) access to all of the data, and coordinate a private computation.

We study this problem both from an information theoretic point of view, and from the point of view of efficient algorithms. We say that an algorithm for the private heavy hitters problem is \emph{efficient} if it runs in time poly$(n, \log N)$: i.e. polynomial in the database size, but only polylogarithmic in the universe size (i.e. in what we view as the most interesting range of parameters, the universe may be exponentially larger than the size of the database). We give tight information theoretic upper and lower bounds on the accuracy to which the heavy hitter can be found in the private distributed setting (separating this model from the private centralized setting), and give several efficient algorithms which achieve good, although information-theoretically sub-optimal accuracy guarantees. We leave open the question of whether \emph{efficient} algorithms can exactly match the information theoretic bounds we prove for the private heavy hitters problem in the distributed setting.

\subsection{Our Results}
In this section, we summarize our results. The bounds we discuss here are informal and hide many of the parameters which we have not yet defined. The formal bounds are given in the main body of the paper.

First, we provide an information theoretic characterization of the accuracy to
which any algorithm (independent of computational constraints) can solve the
heavy hitters problem in the private distributed setting. We say that an
algorithm is $\alpha$-accurate if it returns a universe element which occurs
with frequency at most an additive $\alpha$ smaller than the true heavy hitter.
In the centralized setting, a simple application of the exponential
mechanism~\cite{MT07} gives an $\alpha$-accurate mechanism for the heavy-hitters
problem
where $\alpha = O(\log |N|)$, which in particular, is independent of the number
of individuals $n$. In contrast, we show that in the fully distributed setting,
no algorithm can be $\alpha$-accurate for $\alpha = \Omega(\sqrt{n})$ even in the case in which $|N| = 2$. Conversely, we give an almost matching upper bound (and an algorithm with run-time linear in $N$) which is $\alpha$-accurate for $\alpha = O(\sqrt{n \log N})$.

Next, we consider \emph{efficient} algorithms which run in time only
polylogarithmic in the universe size $|N|$. Here, we give two algorithms. One is
an application of a compressed sensing algorithm of Gilbert et al.~\cite{GLPS09},
which is $\alpha$-accurate for $\alpha = \tilde{O}(n^{5/6} \log^{1/6} N)$.
Then, we give an algorithm based on group-testing using pairwise independent
hash functions, which has an incomparable bound. Roughly speaking, it guarantees
to return the exact heavy hitter (i.e. $\alpha = 0)$ whenever the frequency of
the heavy hitter is larger than the $\ell_2$-norm of the frequencies of the
remaining elements. Depending on how these frequencies are distributed, this can
correspond to a bound of $\alpha$-accuracy for $\alpha$ ranging anywhere between
the optimal $\alpha = O(\sqrt{n})$ to $\alpha = O(n)$.
\subsection{Our Techniques}
Our upper bounds, both information theoretic, and those with efficient
algorithms, are based on the general technique of \emph{random projection} and
\emph{concentration of measure}. To prove our information theoretic upper bound,
we observe that to find the heavy hitter, we may view the private database as a
histogram $v$ in $N$ dimensional space. Then, it is enough to find the index $i
\in [N]$ of the universe element which maximizes $\langle v, e_i\rangle$, where
$e_i$ is the $i$'th standard basis vector. Both $v$ and each $e_i$ have small
$\ell_1$-norm, and so each of these inner products can be approximately
preserved by taking a random projection into $\tilde{O}(\log N)$ dimensional
space. Moreover, we can project each individual's data into this space
independently in the fully distributed setting, incurring a loss of only
$O(\sqrt{n})$ in accuracy. This mechanism, however, is not efficient, because to
find the heavy hitter, we must enumerate through all $|N|$ basis vectors $e_i$
in order to find the one that maximizes the inner product with the projected
database. Similar ideas lead to our efficient algorithms, albeit with worse
accuracy guarantees. For example, in our first algorithm, we apply techniques from compressed sensing to the projected database to recover (approximately) the heavy hitter, rather than checking basis vectors directly. In our second algorithm, we take a projection using a particular family of pairwise-independent hash functions, which are linear functions of the data universe elements. Because of this linearity, we are able to efficiently ``invert'' the projection matrix in order to find the heavy hitter.

Our lower bound separates the distributed setting from the centralized setting by applying an anti-concentration argument. Roughly speaking, we observe that in the fully distributed setting, if individual data elements were selected uniformly $i.i.d.$ from the data universe $N$, then even after conditioning on the messages exchanged with any differentially private algorithm, they remain independently distributed, and approximately uniform. Therefore, by the Berry-Esseen theorem, even after any algorithm computes its estimate of the heavy hitter, the true distribution over counts remains approximately normally distributed. Since the Gaussian distribution exhibits strong anti-concentration properties, this allows us to unconditionally give an $\Omega(\sqrt{n})$ lower bound for any algorithm in the fully distributed setting.
\subsection{Related Work}
Differential privacy was introduced in a sequence of papers culminating in~\cite{DMNS06}, and has since become the standard ``solution concept'' for privacy in the theoretical computer science literature. There is by now a very large literature on this topic, which is too large to summarize here. Instead, we focus only on the most closely related work, and refer the curious reader to a survey of Dwork~\cite{Dwo08}.

Most of the literature on differential privacy focuses on the \emph{centralized}
model, in which there is a trusted database administrator. In this paper, we
focus on the \emph{local} or \emph{fully distributed} model, introduced by
\cite{KLNRS08}. There has been little work in this more
restrictive model--the problems of \emph{learning}~\cite{KLNRS08} and
\emph{query release}~\cite{GHRU11} in the local model are well
understood\footnote{Roughly, the set of concepts that can be \emph{learned} in
the local model given polynomial sample complexity is equal to the set of
concepts that can be learned in the SQ model given polynomial query complexity
\cite{KLNRS08}, and the set of queries that can be \emph{released} in the local
model given polynomial sample complexity is equal to the set of concepts that
can be agnostically learned in the SQ model given polynomial query complexity~\cite{GHRU11},
but the polynomials are not equal.}, but only up to polynomial factors that do
not imply tight bounds for the heavy hitters problem. The \emph{two-party}
setting (which is intermediate between the centralized and fully distributed
setting), in which the data is divided between two databases without a trusted central administrator, was considered by~\cite{MMPRTV10}. They proved a separation between the two-party setting and the centralized setting for the problem of computing the Hamming distance between two strings. In this work, we prove a separation between the fully distributed setting and the centralized setting for the problem of estimating the heavy hitter.

A variant of the private heavy hitters problem has been considered in the
setting of \emph{pan-private streaming algorithms}~\cite{DNPRY10,MMNW11}. This
work considers a different (although related) problem in a different (although
related) setting.~\cite{DNPRY10,MMNW11} consider a setting in which a stream of
elements is presented to the algorithm, and the algorithm must estimate the
\emph{approximate count} of frequently occurring elements (i.e. the number of
``heavy hitters''). In this setting, the universe elements themselves are the
individuals appearing in the stream, and so it is not possible to reveal the
identity of the heavy hitter. In contrast, in our work, individuals are distinct
from universe elements, which merely label the individuals. Moreover, our goal
here is to actually identify a specific universe element which is the heavy
hitter, or which occurs almost as frequently. Also,~\cite{DNPRY10,MMNW11} work in the centralized setting, but demand \emph{pan-privacy}, which roughly requires that the internal state of the algorithm itself remain differentially private. In contrast, we work in the \emph{local privacy} setting which gives a guarantee which is strictly stronger than pan-privacy. Because algorithms in the local privacy setting only interact with individuals in a differentially private way, and never have any other access to the private data, any algorithm in the local privacy model can never have its state depend on data in a non-private way, and such algorithms therefore also preserve pan-privacy. Therefore, our upper bounds hold also in the setting of pan-privacy, whereas our lower bounds do not necessarily apply to algorithms which only satisfy the weaker guarantee of pan-privacy.

Finally, we note that many of the upper bound techniques we employ have been previously put to use in the centralized model of data privacy i.e. random projections~\cite{BLR08,BR11} and compressed sensing (both for lower bounds~\cite{DMT07} and algorithms~\cite{LZWY11}). As algorithmic techniques, these are rarely optimal in the centralized privacy setting. We remark that they are particularly well suited to the fully distributed setting which we study here, because in a formal sense, algorithms in the local model of privacy are constrained to only access the private data using noisy linear queries, which is exactly the form of access used by random linear projections and compressed sensing measurements.

\section{Preliminaries}
\label{sec:prelims}
A database $v$ consists of $n$ records from a data universe $N$, one
corresponding to each of $n$ individuals: for $i \in [n]$, $v^i \in N$ and $v =
\{v^1,\ldots,v^n\}$ which may be a \emph{multiset}. Without loss of generality,
we will index the elements of the data universe from $1$ to $|N|$. It will be
convenient for us to represent databases as \emph{histograms}. In this
representation, $v \in \mathbb{N}^{|N|}$, where $v_i$ represents the number of
occurrences of the $i$'th universe element in the database. Further, we write $v^i \in \mathbb{N}^{|N|}$ for each individual $i \in [n]$, where $v^i_j = 1$ if individual $i$ is associated with the $j$'th universe element, and $v^i_{j'} = 0$ for all other $j' \neq j$. Note that in this histogram notation, we have: $v = \sum_{i=1}^n v^i$. In the following, we will usually use the histogram notation for mathematical convenience, with the understanding that we can in fact more concisely represent the database as a multiset.

Given a database $v$, the \emph{heavy hitter} is the universe element that occurs most frequently in the database: $hh(v) = \arg\max_{i \in N}v_i$. We refer to the frequency with which the heavy hitter occurs as $fhh(v) = v_{hh(v)}$. We want to design algorithms which return universe elements that occur \emph{almost as frequently as the heavy hitter}.
\begin{definition}
An algorithm $A$ is $(\alpha, \beta)$-accurate for the heavy hitters problem if for every database $v \in \mathbb{N}^{|N|}$, with probability at least $1-\beta$: $A(v) = i^*$ such that $v_{i^*} \geq fhh(v) - \alpha$.
\end{definition}

\subsection{Differential Privacy}
Differential privacy constrains the sensitivity of a randomized algorithm to individual changes in its input.
\begin{definition}
An algorithm $A:\mathbb{N}^{|N|}\rightarrow R$ is $(\epsilon,\delta)$-differentially private if for all $v, v' \in \mathbb{N}^{|N|}$ such that $||v-v'||_1 \leq 1$, and for all events $S \subseteq R$:
$$\Pr[A(v) \in S] \leq \exp(\epsilon)\Pr[A(v') \in S] + \delta$$
\end{definition}
Typically, we will want $\delta$ to be negligibly small, whereas we think of $\epsilon$ as being a small constant (and never smaller than $\epsilon = O(1/n)$).

A useful distribution is the \emph{Laplace} distribution.
\begin{definition}[The Laplace Distribution]
The Laplace Distribution (centered at 0) with scale $b$ is the
distribution with probability density function $\textstyle \Lap(x | b)
= \frac{1}{2b}\exp\left( -\frac{|x|}{b}\right)$.
We will sometimes write $\textrm{Lap}(b)$ to denote the Laplace distribution with scale $b$, and will sometimes abuse notation and write $\Lap(b)$ simply to denote a random variable $X \sim \Lap(b)$.
\end{definition}
A fundamental result in data privacy is that perturbing low sensitivity queries with Laplace noise preserves $(\epsilon,0)$-differential privacy.
\begin{theorem}[\cite{DMNS06}]
\label{thm:laplace-privacy}
Suppose $Q : \mathbb{N}^{|N|} \rightarrow \mathbb{R}$ is a function such that for all
databases $v, v' \in \mathbb{N}^{|N|}$ such that $||v-v'||_1 \leq 1$, $|Q(v)-Q(v')|
\leq c$. Then the procedure which on input $v$ releases $Q(v) + X$, where $X$ is a draw from a
$\textrm{Lap}(c/\epsilon)$ distribution, preserves
$(\epsilon,0)$-differential privacy.
\end{theorem}
It will be useful to understand how privacy
parameters for individual steps of an algorithm compose into privacy
guarantees for the entire algorithm.  The following useful theorem is a special case of a theorem proved by
Dwork, Rothblum, and Vadhan:
\begin{theorem}[Privacy Composition~\cite{DRV10}]
\label{thm:compose}
Let $0 < \epsilon, \delta < 1$, and let $M_1,\ldots,M_T$ be $(\epsilon',0)$-differentially private algorithms for some $\epsilon' \leq \epsilon / \sqrt{8 T\log\left(\frac{1}{\delta}\right)}.$
Then the algorithm $M$ which on input $v$ outputs $M(v) = (M_1(v), \ldots, M_T(v))$ is $(\epsilon, \delta)$-differentially private.
\end{theorem}

The local privacy model (alternately, the fully distributed setting) was introduced by Kasiviswanathan et al.
\cite{KLNRS08} in the context of learning. The local privacy model
formalizes  randomized response: there is no central database of
private data. Instead, each individual $i$ maintains possession of their
own data element (i.e. a database $v^i$ of size $||v^i||_1 = 1$), and answers questions about
it only in a differentially private manner. Formally, the database $v
\in \mathbb{N}^{|N|}$ is the sum of $n$ databases of size $1$: $v = \sum_{i=1}^nv^i$, and each $v^i$ is held by individual $i$.
\begin{definition}[\cite{KLNRS08} (Local Randomizer)]
An $(\epsilon,\delta)$-local randomizer $R:\mathbb{N}^{|N|}\rightarrow R$ is an
$(\epsilon,\delta)$-differentially private algorithm that takes a database of
size $||v||_1=1$.
\end{definition}
In the local privacy model, algorithms may interact with the database
only through a local randomizer oracle:
\begin{definition}[\cite{KLNRS08} (LR Oracle)]
An LR oracle $LR_v(\cdot,\cdot)$ takes as input an index $i \in [n]$
and an $(\epsilon,\delta)$-local randomizer $R$ and outputs a random value $w \in
R$ chosen according to the distribution $R(v^i)$, where $v^i$ is
the element held by the $i$'th individual in the database.
\end{definition}

\begin{definition}[\cite{KLNRS08} (Local Algorithm)]
An algorithm is $(\epsilon,\delta)$-local if it accesses the database $v$ via the
oracle $LR_v$, that satisfies the following restriction: if
$LR_v(i,R_1),\ldots,LR_v(i,R_k)$ are the algorithm's invocations of
$LR_v$ on index $i$, then the joint outputs of each of these $k$ algorithms must be $(\epsilon,\delta)$-differentially private.
\end{definition}

To avoid cumbersome notation, we will avoid the formalism of LR oracles, instead remembering that for algorithms in the local model, any operation on $v^i$ must be carried out without access to any $v^j$ for $j \neq i$, and must be differentially private in isolation.

\subsection{Probabilistic Tools}
We will make use of several useful probabilistic tools. First, the well-known Johnson-Lindenstrauss lemma:

\begin{theorem}[Johnson-Lindenstrauss Lemma]
\label{lemma-jl}
Let $0 < \gamma < 1$ be given. For any set $V$ of $q$ vectors in $\RR^N$, there
exists a linear map $A : \RR^N \rightarrow \RR^m$ with $m = O\left(\frac{\log
q}{\gamma^2}\right)$ such that $A$ is approximately an isometric embedding of
$V$ into $\RR^m$. That is, for all $x, y \in V$, we have the two bounds
\[ (1 - \gamma) \|x - y\|^2 \leq \|A(x - y)\|^2 \leq (1 + \gamma) \|x - y\|^2 \]
\[ |\langle Ax, Ay\rangle - \langle x,y\rangle | \leq O(\gamma (\|x\|^2 + \|y\|^2)) \]
In particular, any $m \times N$ random projection matrix $A_p$, whose entries are drawn IID
uniformly from $\{-1/\sqrt{m}, 1/\sqrt{m}\}$, enjoys this property with
probability at least $1-\beta$, with $m = O\left(\frac{\log q \log
(1/\beta)}{\gamma^2}\right)$. Note that this projection matrix does not depend
on the set of vectors $V$.
\end{theorem}

In other words, any set of $q$ points in a high dimensional space can be {\em obliviously} embedded into a space of dimension $O(\log q)$ such that w.h.p. this embedding essentially preserves pairwise distances.

In our analysis, we will also make use of a simple tail bound on the sums of Laplace random variables:

\begin{theorem}[See, e.g.~\cite{GRU12}]
\label{thm-laplace-sum}
Let $X_i, i \in [n]$ be IID random variables drawn from the $\Lap(b)$ (the
Laplace distribution with parameter $b$) and let $X = \sum_{i = 1}^n X_i$.
Then, we have the bound

\begin{displaymath}
  \Pr [ X \geq T ] \leq \left\{
    \begin{array}{lr}
      \exp\left(-\frac{T^2}{6 n b^2}\right) & : T \leq nb \\
      \exp\left(-\frac{T}{6 b}\right) & : T > n b
    \end{array}
   \right.
\end{displaymath}

In particular, choosing $T_\beta = b \sqrt{6n} \log(2/\beta)$ gives
\[ \Pr [ |X| \leq T_\beta ] \geq 1 - \beta \]
\end{theorem}

\section{Information Theoretic Upper and Lower Bounds.}
In this section we present upper and lower bounds on the accuracy to which any algorithm in the fully distributed model can privately approximate heavy hitters. Our upper bound can be viewed as an algorithm, albeit one that runs in time linear in $|N|$ and so is not what we consider to be efficient.

\subsection{An Upper Bound via Johnson-Lindenstrauss Projections}
\label{sec-jl}

We present here our first algorithm, referred to as {\em JL-HH}, that solves the heavy hitters problem in the local model using the Johnson-Lindenstrauss lemma. The algorithm JL-HH is outlined in \Cref{alg-JL-HH}. We write
$e_i$ to refer to the $i$'th standard basis vector in $\RR^N$, and write
$\mbox{RandomProjection}(m,N+1)$ for a subroutine which returns a linear embedding of $N+1$ points
into $m$ dimensions using a random $\pm 1/\sqrt{m}$ valued projection matrix, as specified by the Johnson-Lindenstrauss lemma. By the Johnson-Lindenstrauss lemma, for any set of $N+1$ elements, this map
approximately preserves pairwise distances with high probability.

\begin{algorithm}
\caption{JL-HH Mechanism}
\label{alg-JL-HH}
\begin{algorithmic}
\REQUIRE Private histograms $v^{i} \in \NN^N, i \in [n]$. Privacy parameters
$\epsilon, \delta > 0$. Failure probability $\beta > 0$.
\ENSURE $p^*$, index of the heavy hitter.

\STATE $\gamma \gets 1/n^2$
\STATE $m \gets \frac{\log (N+1) \log(2/\beta)}{\gamma^2}$
\STATE $A \gets \mbox{RandomProjection}(m, N+1)$

\FOR{$p = 1$ \mbox{to} $N$ indices}
\FOR{$i = 1$ \mbox{to} $n$ users}
\STATE $z^{i} \sim
\left\{\Lap\left(\frac{\sqrt{8\log(1/\delta)}}{\epsilon}\right)\right\}^m$
\STATE $q^{i} = A v^{i} + z^{i}$
\STATE $r_{ip} = \langle Ae_p, q^{i} \rangle$
\ENDFOR

\STATE $c_p \gets \sum_{i = 1}^n r_{ip}$
\ENDFOR

\STATE $p^* \gets \argmax_p c_p$
\RETURN $p^*$
\end{algorithmic}
\end{algorithm}

JL-HH is based on the following straightforward idea. If $v$ is a private histogram, we will estimate the
count of the $i$'th element ($\langle v, e_i \rangle$), by estimating $\langle
Av, Ae_i \rangle$, and returning the largest count. By \Cref{lemma-jl}, since we
are using the random projections matrix, we have that with high probability,
inner products between points in the set $V = \{e_1 \cdots e_N, v\}$ are
approximately preserved under $A$. However, we cannot access $Av$ directly since $v$
is private data. To preserve differential privacy, our mechanism must add noise
$z$ to $Av$, and work only with the noisy samples.  Our analysis will thus focus on
bounding the error introduced by this noise term.  First, though, we show that
JL-HH is differentially private.

\begin{lemma}
JL-HH operates in the local privacy model and is $(\epsilon,\delta)$-differentially private.
\end{lemma}
\begin{proof}
The measurement $Av$ is computed in the fully distributed setting, by computing $Av \approx \sum_{i=1}^nAv^i + z^i$. Each individual $i$ may compute $Av^i + z^i$ which corresponds to answering a sequence of $m$ linear queries, each
with sensitivity $1/\sqrt{m}$. By \Cref{thm:laplace-privacy}, the noise that JL-HH
adds guarantees that each such query is $\epsilon_0$-differentially private,
with
\[\epsilon_0 = \frac{\epsilon}{\sqrt{8m \log(1/\delta)}} \]
Thus, by \Cref{thm:compose}, this composition is $(\epsilon,
\delta)$-differentially private, as desired. From here, the algorithm works with
the noised measurement instead of private data, and is therefore differentially
private.
\end{proof}

Now, we show that JL-HH estimates the counts to within an additive error of
$O\left(\frac{\sqrt{n \log N}}{\epsilon}\right)$.

\begin{theorem}
For any $\beta > 0$, JL-HH mechanism is $(\alpha, \beta)$-accurate for the heavy
hitters problem, with $\alpha = O\left(\frac{\sqrt{n\log
(N/\beta)\log(1/\delta)}}{\epsilon}\right)$.
\end{theorem}
\begin{proof}
Let $v$ be the private histogram, and let $z = \sum_{i=1}^n z^i$ denote the sum of the noise vectors added to each individual's data $v^i$. The error of the mechanism is at most
$$2\max_{i\in [N]}\left| \langle e_i, v\rangle - \langle Ae_i, Av + z\rangle\right|$$
Note that for all $j$, the random variable $z_j$ is distributed as the sum of $n$ i.i.d. Laplace random variables each with scale $b = \sqrt{8\log1/\delta}/\epsilon$. To calculate the
error for an index $i$, we may write:
\begin{align}
|\langle e_i, v \rangle - \langle Ae_i, Av + z\rangle| &\leq |\langle e_i, v
\rangle - \langle Ae_i, Av \rangle| + |\langle Ae_i, z \rangle| \\
&= O(\gamma \|v\|^2 + | \langle Ae_i, z \rangle |) \label{jl-approx-bd}
\end{align}
with the second equality following from \Cref{lemma-jl}.
Recall that we have set $\gamma = n^{-2}$, and let $A$ be the random projection matrix, with $m = O(
\log N \log(2/\beta)/ \gamma^2)$.  With probability at least $1 - \beta/2$, the
random projections matrix $A$ actually satisfies the property for the
Johnson-Lindenstrauss lemma. So, we have
\[ \langle Ae_i, z \rangle = \sum_{j=1}^m (Ae_i)_j \sum_{i=1}^n
z^i_j \]
But $Ae_i$ is a vector of length $m$ with entries drawn uniformly from $\pm
1/\sqrt{m}$.  Since the Laplace distribution is also symmetric, the distribution
of this sum is identical to the distribution of a sum of $mn$ i.i.d. Laplace random
variables each with scale $b = \frac{\sqrt{8\log1/\delta}}{\sqrt{m}\epsilon}$. By our tail bound in \Cref{thm-laplace-sum}, with
probability at least $1 - \beta / 2N$, this sum is bounded by
$O\left(\frac{\sqrt{n\log(1/\delta)\log (N/\beta)}}{\epsilon}\right)$.

On the other hand, the other error $|\langle e_i, v \rangle - \langle Ae_i, Av
\rangle|$ can be bounded by \Cref{jl-approx-bd}, and hence is $O(1)$
by our choice of $\gamma$. Thus, with probability at least $1 - \beta / 2N$, we
have that the estimated count for index $i$ is within an additive factor of
$O\left(\frac{\sqrt{n\log (1/\delta) \log (N/\beta)}}{\epsilon}\right)$ to the
true count of index $i$.  Taking a union bound over all indices, we have that
with probability at least $1 - \beta/2$, this accuracy holds for the heavy
hitter, and all other elements.  Since the probability of failing when picking
$A$ was at most $\beta / 2$, this gives the desired high probability bound.
\end{proof}

It is worthwhile to compare JL-HH with a more naive approaches. A simpler
differentially private algorithm to solve the distributed heavy hitters problem
is to have each user simply add noise $\Lap(1/\epsilon)$ to each entry in the
user's private histogram, and report this vector to the central party, which
sums the noisy vectors and estimates the most frequently occurring item.  This
is differentially private, as any neighboring histogram will change exactly one
entry in a user's histogram. However, this method requires having each user
transmit $O(N)$ amount of information to the central party. JL-HH achieves
similar accuracy compared to this naive approach, but since the clients compress
the histogram first, only $O(\log N)$ information must be communicated. Even
though JL-HH runs in time linear in $N$, there are natural situations where long
running time can be tolerated, but large communication complexity cannot, for
instance, if the central party is a server farm with considerable computational
resources, but the communication with users must happen over standard network
links.

\subsection{A Lower Bound via Anti-Concentration}
\label{sec-lowerbound}

Here we show that our upper bound in the previous subsection is essentially optimal: for any $\epsilon < 1/2$ and any $\delta$ bounded away from $1$ by a constant, no $(\epsilon,\delta)$-private mechanism in the fully distributed setting can be $\alpha$-accurate for the heavy hitters problem for some $\alpha = \Omega(\sqrt{n})$, even in the case in which $|N| = 2$. Our theorem follows by arguing that even after conditioning on the output of the differentially private interaction with each individual in the local model, there is still quite a bit of uncertainty in the distribution over heavy hitters, if the universe elements were initially distributed uniformly at random. We take advantage of this uncertainty to apply an anti-concentration argument, which implies that no matter what answer the algorithm predicts, there is enough randomness leftover in the database instance that the algorithm is likely to be incorrect (with at least some constant probability $\beta$). We remark that our technique (while specific to the local privacy model) holds for $(\epsilon,\delta)$-differential privacy, even when $\delta > 0$. This is similar to lower bounds based on reconstruction arguments~\cite{DN03}, and in contrast to other techniques for proving lower bounds in the centralized model, such as the elegant packing arguments used in~\cite{BKN10,HT10}, which are specific to $(\epsilon,0)$-differential privacy. We use an independence argument also used by~\cite{MMPRTV10} to prove a lower bound in the two-party setting.

\begin{theorem}
\label{thm-lower-bound}
For any $\epsilon \leq 1/2$ and $\delta < 1$ bounded away from $1$, there exists an $\alpha = \Omega(\sqrt{n})$ and a $\beta = \Omega(1)$ such that no $(\epsilon,\delta)$-private mechanism in the local model is $(\alpha, \beta)$-accurate for the heavy hitters problem.
\end{theorem}
\begin{proof}
We give a lower bound instance in which the universe is $N = \{0,1\}$. Each individual $i$ is assigned a universe element $s_i \in \{0,1\}$ uniformly at random.  Let $A_i:N\rightarrow \cM$ denote the $(\epsilon,\delta)$-differentially private algorithm which acts on the data $s_i$ of individual $i$, and write $m_i = A_i(s_i)$.

We condition on the order of the parties that we query and on the output
of each algorithm, $m_i = \hat{m_i}$ for fixed $\hat{m_i} \in \cM$.

We first observe that conditioning on the outputs of each $A_i$: $m_i = \hat{m}_i$ for each $i$, the random variables $s_i$ remain independent of one another. (This is a standard fact from communication complexity)

We next argue that under this conditioning, the marginal distributions of a constant fraction of the $s_i$ variables remain
approximately uniform. If we define the random variables $X_i$ to be the
indicator of the event $s_i = \hat{s_i}$ (conditioning on all the messages), we
can apply Bayes' rule to get for all $i \in [n]$:
\begin{align*}
\Pr[X_i = \hat{s_i}] &= \Pr[s_i = \hat{s_i} | m_i = \hat{m_i}] \\
&= \frac{\Pr[m_i = \hat{m_i} | s_i = \hat{s_i}] \Pr[s_i = \hat{s_i}]}{\Pr[m_i =
\hat{m_i}]} \\
&\leq \frac{\Pr[m_i = \hat{m_i} | s_i = \hat{s_i}] \Pr[s_i = \hat{s_i}]}{\Pr[m_i =
\hat{m_i} | s_i = b]}
\end{align*}
where $b$ is some element of the universe. Because each $A_i$ is $(\epsilon,\delta)$-differentially private, we have that with probability at least $1-\delta$, the following random variable (where the randomness is over the choice of $\hat{m}_i$) is bounded:
\[ \frac{\Pr[m_i = \hat{m_i} | s_i = \hat{s_i}]}{\Pr[m_i = \hat{m_i} | s_i =
b]} \leq e^\epsilon\]
and thus with probability $1-\delta$ over the choice of $\hat{m}_j$: $\Pr[X_i = \hat{s_i}] \leq (e^\epsilon)  / 2$, using the prior on $s_i$.

In similar fashion, we can prove a lower bound on the probability. So, we have that for each $i$ independently with probability at least $1-\delta$:
$\Pr[X_i = \hat{s_i}] \in [(e^{-\epsilon})/2, (e^\epsilon) / 2]$. Because we assume $\epsilon \leq 1/2$, we therefore have for each $i$ independently with probability $1-\delta$: $\Pr[X_i = \hat{s_i}] \in [c_1,c_2]$ where $c_1,c_2$ are constants bounded away from $0$ and $1$ respectively. Because this occurs with constant $1-\delta$ probability for each $i$, for any constant $\beta$, we can (by the Chernoff bound) take $n$ to be sufficiently large so that except with probability $\beta/2$, we have $\Pr[X_i = \hat{s_i}] \in [c_1,c_2]$ for $\Omega(n)$ individuals $i$. This, together with the conditional independence of the $X_i$'s, allows us to apply the Berry-Esseen theorem:

\begin{theorem}[Berry-Esseen]
Given independent random variables $X_i, i \in [n]$, let $\mu_i = \EE[X_i],
\sigma_i^2 = \EE[(X_i - \mu_i)^2], \beta_i = \EE[|X_i - \mu_i|^3]$, and let
\[S_n = \frac{ \sum_{i=1}^n (X_i - \mu_i) }{\sqrt{\sum_{i = 1}^n \sigma_i^2}} \]
If $F_n$ is the cdf of $S_n$, and $\Phi$ is the cdf for the standard normal
distribution, then there exists a constant $C$ such that
\[ \sup_x |F_n(x) - \Phi(x)| \leq C \psi \]
where
\[ \psi = \left(\sum_{i = 1}^n \sigma_i^2\right)^{-1/2} \max
\frac{\beta_i}{\sigma_i^2} \]
\end{theorem}

For each of the $\Omega(n)$ individuals $i$ for which $\Pr[X_i = 1] \in [c_1,c_2]$, each $\sigma_i^2$ and $\beta_i$ is a constant bounded away from $0$. Thus, we have with probability at least $\beta/2$: $\psi \leq O(1/\sqrt{n})$, and hence the cdf $F_n$
of the sample mean $S_n$ converges uniformly to the normal distribution. By a
change of variables, this means that the cdf of the sum $\sum_{i=1}^n (X_i -
\mu_i)$ converges to the cdf of a normal distribution with mean $0$ and variance
$\sigma^2 = \sum_{i = 1}^n \sigma_i^2 = \Omega(n)$. The next lemma lower
bounds the probability that $S_n$ is within an additive factor of $\Omega(\sqrt{n})$ of its mean.

\begin{lemma}
Let $\beta > 0$ be given and condition on the event that $\Pr[X_i = 1] \in [c_1,c_2]$ for $\Omega(n)$ individuals $i \in [n]$.  For sufficiently large $n$, there exists a constant $C$ such that
\[ \Pr\left[ \left|\sum_{i = 1}^n (X_i - \mu_i)\right| \geq C \sqrt{n}\right]
\geq 1 - \beta/2 \]
\end{lemma}
\begin{proof}[Proof of Lemma]
This is immediate, since by the Berry-Esseen theorem the sum $\sum_{i = 1}^n (X_i - \mu_i)$ converges uniformly to a Gaussian distribution with standard deviation $\sigma = \Omega(\sqrt{n})$.
%
%By the Berry-Esseen theorem, the distribution of $\sum_{i = 1}^n (X_i - \mu_i)$
%converges to the normal distribution $\cN(0, \sigma^2)$.  Recall that we can
%bound $\sigma \geq B\cdot \sqrt{n})$ for some constant $B$. For $X \sim \cN(0,1)$,
%drawn from the standard normal distribution, pick $C$ such that
%\[ \Pr[ |X| \geq C/B ] \leq \delta \]
%
%Now, we can calculate:
%
%\[ \int_{y = C \sqrt{n}}^\infty \frac{1}{\sigma \sqrt{2\pi}} e^{-y^2 /
%2\sigma^2} dy = \int_{y' = C\sqrt{n}/\sigma}^\infty \frac{1}{\sqrt{2\pi}} e^{-y'^2 / 2}
%dy' \]
%
%But since $\sigma \geq B \sqrt{n}$, and since the right hand side is
%just the cdf of the standard normal distribution, we are done by our choice of
%$C$.
\end{proof}

To complete the proof, we note that the distribution of $n_1 \equiv \sum_{i=1}^n X_i$ is simply the
distribution of the number of occurrences of universe element $1$, after conditioning on the outcome of differentially private mechanisms $A_1,\ldots,A_n$. Consider a mechanism, which given the outcome of mechanisms $A_1,\ldots,A_n$ attempts to guess the value of $n_1$, and outputs $\hat{n}_1$. Let $\mu = \sum_{i=1}^n \mu_i$. By the properties of the Gaussian distribution we have:
$$\Pr[|n_1 - \hat{n_1}| \leq t] \leq \Pr[|n_1 - \mu|  \leq t]$$
for all values of $t$. In particular, for some $t = C \sqrt{n}$ we have shown that this probability is at most $\beta$. In other words, we have shown that for some constant $\beta \geq 0$ and for some $\alpha = \Omega(\sqrt{n})$, there is no $(\epsilon,\delta)$-private algorithm in the local model which is able to estimate the \emph{frequency} of the heavy hitter to within an additive $\alpha$ factor with probability $1-\beta$. It is straightforward to see that there therefore cannot be an $(\alpha,\beta)$-accurate, $(\epsilon,\delta)$-private mechanism for the heavy hitters problem: any such mechanism could be converted to a mechanism which estimates the frequency of the heavy hitter by introducing ``dummy'' individuals corresponding to the universe element which is not the heavy hitter, and performing a binary search over their count by computing the identity of the heavy hitter in each dummy instance. The count at which the identity of the heavy hitter in the dummy instance changes can then be used to estimate the frequency of the true heavy hitter.

 %On the other hand, recall that each $s_i$ was initially selected uniformly at random, and so by the properties of the Binomial distribution, for every constant $k$, we have that with constant probability: $n_1 \in [n/2- k\sqrt{n}, n/2 + k\sqrt{n}]$. Setting $k < C/2$, we have that when this event occurs, any algorithm which identifies the heavy hitter can be used to estimate $n_1$ to accuracy better than $C \sqrt{n}$, which we have shown cannot occur with probability better greater than $\beta$.

%Therefore, for some $\alpha = \Omega(\sqrt{n})$, there is a constant $\beta$ such that there does not exist any $(\epsilon,\delta)$-private mechanism in the local model which is $(\alpha,\beta)$-accurate, which completes the proof.
\end{proof}

\section{Efficient Algorithms}
In the last section, we saw the Johnson-Lindenstrauss algorithm which gave almost optimal accuracy guarantees, but had running time linear in $|N|$. In this section, we consider efficient algorithms with running time poly$(n, \log |N|)$. The first is an application of a sublinear time algorithm from the compressed sensing literature, and the second is a group-testing approach made efficient by the use of a particular family of pairwise-independent hash functions.
\subsection{GLPS Sparse Recovery}
\label{sec-strauss}

In this section we adapt a sophisticated algorithm from compressed
sensing. Gilbert, et al.~\cite{GLPS09} present a sparse recovery algorithm (we
refer to it as the GLPS algorithm) that takes linear measurements from a sparse vector,
and reconstructs the original vector to high accuracy. Importantly, the algorithm runs in time
polylogarithmic in $|N|$, and polynomial in the sparsity parameter of the
vector. We remark that our database $v$ is $n$-sparse: it has at most $n$ non-zero components. In the rest of this section, we will write $v_s$ to denote the vector $v$ truncated to contain only its $s$ largest components.

Let $s$ be a sparsity parameter, and let $\gamma$ be a tunable approximation
level. The GLPS algorithm runs in time $O((s/\gamma) \log^{c} N)$, and makes $m
= O(s \log(N/s) /\gamma))$ measurements from a specially constructed
(randomized) $\{-1,0,1\}$ valued matrix, which we will denote $\Phi$. Given
measurements $u = \Phi v + z$ (where $z$ is arbitrary noise), the algorithm
returns an approximation $\hat{v}$, with error guarantee
\begin{equation}
\label{eq-strauss-bound}
\|v - \hat{v}\|_{2} \leq (1 + \gamma) \|v - v_{s}\|_{2}
	+ \gamma \log(s) \frac{\|z\|_{2}}{\kappa}
\end{equation}
with probability at least $3/4$, where $\kappa = O(\log(s) \log^{1/2}(N / s))$.
Though the GLPS bound only occurs with probability $3/4$, the success
probability can be made arbitrarily close to $1$ by running this algorithm
several times. In particular, using the amplification lemma from~\cite{GLMRT10},
the failure probability can be driven down to $\beta$ at a cost of only a factor
of $\log(1/\beta)$ in the accuracy. In what follows, we analyze a single run of
the algorithm, with $\gamma = 1$.
%For simplicity, we will fix $\gamma = 1$.

\begin{algorithm}
\caption{GLPS-HH Mechanism}
\label{alg-glps}
\begin{algorithmic}
\REQUIRE Private histograms $v^{i} \in \NN^N, i \in [n]$. GLPS matrix $\Phi$.
Privacy parameters $\epsilon, \delta > 0$.
\ENSURE $p^*$, estimated index of heavy hitter.
\STATE $m \gets s \log(N / s)$
\STATE $b \gets \sqrt{8 m \log(1/\delta)} / \epsilon$
\FOR{$i = 1$ \mbox{to} $n$ users}
\STATE $z^{i} \sim \{\Lap(b)\}^m$
\STATE $q^{i} \gets \Phi v^{i} + z^{i}$
\ENDFOR

\STATE $c \gets \sum_{i = 1}^n q^{i}$
\STATE $\hat{v} \gets GLPS(c, \Phi)$
\STATE $p^* \gets \argmax_p \hat{v}_p$
\RETURN $p^*$
\end{algorithmic}
\end{algorithm}

Next, we will show that GLPS-HH is $(\epsilon, \delta)$-differentially private.

\begin{theorem}
GLPS-HH operates in the local privacy model and is $(\epsilon,
\delta)$-differentially private.
\end{theorem}
\begin{proof}
The algorithm operates in the local privacy model because each individual $i$
compute $\Phi v^i + z^i$ independently, which corresponds to answering $m$
linear queries, each with sensitivity $1$. The magnitude of the Laplace noise
added, $z^i$, is then sufficient (by \Cref{thm:compose}) to guarantee
$(\epsilon,\delta)$-differential privacy for each individual.
\end{proof}

Next, we will bound the error that we introduce by adding noise for
differential privacy.

\begin{theorem}
Let $\beta > 0$ be given. GLPS-HH is $(\alpha, 3/4-\beta)$-accurate for the heavy
hitters problem, with \[\alpha =\tilde{O}\left(\frac{n^{5/6} \log^{1/3}
      (1/\beta) \log^{1/6} N \log^{1/6}(1/\delta)}{\epsilon^{1/3}}\right)\]
\end{theorem}
\begin{proof}
Let $b = \sqrt{8m \log(1/\delta)}/\epsilon$.  Let $v$ denote the combined
private database, and let $\hat{v}$ denote the estimated private database
returned by GLPS. GLPS-HH uses the GLPS algorithm with measurements $c = \Phi v
+ z, z = \sum_i z^{i}$, where the noise vector $z$ has each entry drawn from
$\sum_{i = 1}^n \Lap(b)$. From \Cref{thm-laplace-sum}, we have the bound
(for a fixed index $i$)
\[ \Pr[ |z_i| \leq O(b \sqrt{n} \log(m/\beta)) ] \geq 1 - \beta / m \]
Taking a union bound over all $m$ indices, we find this bound holds over all
components with probability at least $1 - \beta$. Thus we can bound
\[
  \|z\|_2 \leq O(b \sqrt{nm} \log(m/\beta)) = O\left(\frac{s \log(N/s)
\log (s\log (N/s)/\beta) \sqrt{n \log(1/\delta)}}{\epsilon}\right)
\]
With probability $3/4$, we have the GLPS bound \Cref{eq-strauss-bound} with
$\gamma = 1$, from which we can estimate
\[
  \|v - \hat{v}\|_\infty \leq \|v - \hat{v}\|_2 \leq 2\|v - v_s\|_2 +
O \left(\frac{s \log^{1/2}(N/s) \log(s \log(N/s)/\beta) \sqrt{n
\log(1/\delta)}}{\epsilon}\right)
\]
By a Lemma from~\cite{GSTV07}, we have $\|v - v_s\|_2 \leq \|v\|_1 / \sqrt{s}$.
Now, in the worst case, $\|v\|_1 = O(n)$, and we need to choose $s$ to balance
the errors in
\[
  \|v - \hat{v}\|_\infty \leq 2\frac{n}{\sqrt{s}} + O \left(\frac{s
      \log^{1/2}(N/s) \log(s \log (N/s)/\beta) \sqrt{n
        \log(1/\delta)}}{\epsilon}\right)
\]
Setting $s$ to be
\[ s = \left( \frac{\epsilon}{\log(1/\beta)} \sqrt{\frac{n}{\log N \log
        (1/\delta)}} \right)^{2/3} \]
we get an error bound
\[ \|v - \hat{v}\|_\infty \leq \tilde{O}\left(\frac{n^{5/6} \log^{1/3}
      (1/\beta)\log^{1/6} N \log^{1/6}(1/\delta)}{\epsilon^{1/3}}
  \right)\]
where we write only the dominant factor for each variable (hiding $\log \log$
factors and $\log \epsilon, \log n$ factors).  Thus, with probability at least
$3/4 - \beta$, we get the desired accuracy.
\end{proof}

\subsection{The Bucket mechanism}
\label{sec-bucket}
In this section we present a second computationally efficient algorithm, based on group-testing and a specific family of pairwise independent hash functions.

%In \Cref{alg-bucket}, we present
%the \emph{Bucket mechanism}. This algorithm relies on pairwise-independent hash
%functions, and seeks to find the heavy hitter without estimating the whole
%histogram.

\begin{algorithm}
\caption{The Bucket Mechanism}
\label{alg-bucket}
\begin{algorithmic}
\REQUIRE Private labels $v^i \in [N], i \in [n]$. Failure probability $\beta >
0$. Privacy parameters $\epsilon, \delta > 0$.
\ENSURE $p^*$, the index of the heavy hitter.

\STATE $F \gets \{0, 1\}^{\log N} \setminus 0$
\FOR {$i = 1$ to $8 \log (1 / \beta)$ trials}
\STATE $H \in \{0, 1\}^{\log(12N) \times \log N} \gets $ Draw $\log(12N)$ rows from $F$,
uniformly at random.

\STATE $u \in \RR^{\log(12N)} \gets 0$
\FOR {$j = 1$ to $n$ users}
\STATE $b \in \{0, 1\}^{\log N} \gets $ binary expansion of $v^j$.
\STATE $s \gets Hb \pmod{2}$
\STATE $z \sim \left\{\Lap\left(\frac{8 \sqrt{\log(12N) \log(1/\beta) \log
(1/\delta)}}{\epsilon}\right)\right\}^{\log(12N)}$
\STATE $u \gets u + s + z$
\ENDFOR

\FOR {$k = 1$ to $\log(12N)$ hash functions}
\STATE $b_k \gets
\left\{
  \begin{array}{ll}
    1 &: u_k > n/2 \\
    0 &: \mbox{otherwise}
  \end{array}
\right.$
\ENDFOR

\STATE $w_i \gets
\left\{
  \begin{array}{ll}
    x_0 &: Hx_0 = b \pmod{2} \\
    \perp &: Hx = b \pmod{2} \mbox{\,infeasible}
  \end{array}
\right.$
\ENDFOR

\STATE $w^* \gets $ most frequent $w_i$, ignoring $\perp$
\RETURN $p^* \gets w^*$ converted from binary

\end{algorithmic}
\end{algorithm}

At a high level, our algorithm, referred to as the {\em Bucket mechanism}, runs
$O(\log(1/\beta))$ \emph{trials} consisting of $O(\log N)$ 0/1 valued hash
functions in each trial. For a given trial, the mechanism hashes each universe
element into one of two buckets for each hash function. Then, the mechanism
tries to find an element that hashes into the bucket with more weight (the
\emph{majority bucket}) for all the hash functions. If there is such an element,
it is a candidate for the heavy hitter for that trial. Finally, the mechanism
takes a majority vote over the candidates from each trial to output a final
heavy hitter.

For efficiency purposes we do not use truly random hash functions, but instead
rely on a particular family of pairwise-independent hash functions which can be
expressed as linear functions on the bits of a universe element. Specifically,
each function $h$ in the family maps $[N]$ to $\{0,1\}$, and is parameterized by
a bit-string $r \in \{0,1\}^{\log |N|}$. In particular, given any bit-string $r
\in \{0,1\}^{\log |N|}$, we define $h_r(x) = \langle r, b(x) \rangle$, where
$b(x)$ denotes the binary representation of $x$. $r$ is chosen uniformly at
random from the set of all strings $r \in \{0,1\}^{\log |N|}\setminus 0^{\log
|N|}$. Given hash functions of this form, and a list of target buckets, the
problem of finding an element that hashing to all of the target buckets is
equivalent to solving a linear system mod $2$, which can be done efficiently.
Our family of hash functions operates on the element label in binary, hence the
conversions to and from binary in the algorithm.

We will now show that the bucket mechanism is $(\epsilon,\delta)$-differentially
private, runs in time poly$(n, \log |N|)$, and assuming a certain condition on
the distribution over universe elements, returns the exact heavy hitter. The
accuracy analysis proceeds in two steps: first, we argue that with constant
probability $> 1/2$, the heavy hitter is the unique element hashed into the
larger bucket by every hash function in a given trial. Then, we argue that with
high probability, the proceeding event indeed occurs in the majority of trials,
and so the majority vote among all trials returns the true heavy hitter.

\begin{theorem}
The Bucket mechanism operates in the local model and is $(\epsilon, \delta)$-differentially private.
\end{theorem}
\begin{proof}
Each party answers $\log(12N)$ $1$-sensitive queries about only their own data for each trial, with a total
of $8 \log(1/\beta)$ trials. By \Cref{thm:compose}, the correct amount of
noise is added to preserve $(\epsilon, \delta)$-differential privacy.
\end{proof}

\begin{theorem}
For fixed $\epsilon, \delta > 0$ and failure probability $\beta > 0$, the Bucket
mechanism runs in time $O(n\log(1/\beta)\log^3 N)$.
\end{theorem}
\begin{proof}
The step that dominates the run time is the inner loop over each party. For each
user, the algorithm evaluates $O(\log N)$ hash functions. Each evaluation
calculates the inner product of two $\log N$-length bit strings, and there are
$O(\log N)$ hash functions. So, each user takes time $\log^2 N$ per trial. With
$n$ users and $O(\log(1/\beta))$ trials, the result follows.
\end{proof}

We first prove a simple tail bound on sums of
$k$-wise independent random variables, modifying a result given by Bellare and
Rompel~\cite{BR94}.

\begin{lemma}
\label{lemma-k-wise-azuma}
Let $k$ be even. Take a $k$-independent set of random variables $X_i$, with
$0 \leq X_i \leq c_i$, let $X = \sum X_i$, and let $\mu = \EE[X]$.
We have:
\[ \Pr[|X - \mu| > t] \leq C_k \left(\frac{c k}{t^2}\right)^{k/2} \]
with $c = \sum c_i^2$, and $C_k = 2 \sqrt{\pi k}e^{k/2 - 1/(6k)} \leq 1.0004$.
\end{lemma}
\begin{proof}
By Markov's inequality, we can write:
\[ \Pr[|X - \mu| > t] = \Pr[(X - \mu)^k > t^k] \leq \frac{\EE[(X - \mu)^k]}{t^k}
\]
However, if we expand out the product, we find that we only need to consider the
expected value of products of at most $k$ of the variables $X_i$. Thus, without
loss of generality, we may consider $X_i$ to be independent for the following
calculation.

\begin{align*}
\EE[(X - \mu)^k] &= \int_0^\infty \Pr[(X - \mu)^k > s] ds \\
&= \int_0^\infty \Pr[|X - \mu| > s^{1/k}] ds \\
&\leq \int_0^\infty 2 \exp\left(-\frac{s^{2/k}}{2 \sum c_i^2}\right) ds
\end{align*}
where we have used that the $X_i$ are independent in order to applied Azuma's
inequality.  By a change of variables, and letting $c = \sum c_i^2$, we have

\begin{align*}
\EE[(X - \mu)^k] &\leq \int_0^\infty k (2c)^{k/2} e^{-y} y^{k/2 - 1} dy \\
&= (2c)^{k/2} k \Gamma(k/2 - 1) \\
&= 2(2c)^{k/2} \left(\frac{k}{2}\right)! \\
&\leq 2c^{k/2} \sqrt{\pi k} \left(\frac{k}{e}\right)^{k/2} e^{1/6k} \\
\end{align*}
where we have used Stirling's approximation in the last step. Now, we get
\[ \Pr[|X - \mu| > t] = \frac{\EE[(X - \mu)^k]}{t^k} \leq C_k \left(\frac{c
k}{t^2}\right)^{k/2} \]
as desired.
\end{proof}

\begin{lemma}
\label{lemma-heavy-wins}
Let $\beta, \epsilon, \delta > 0$ be given, and consider a single trial in the
Bucket mechanism. Without loss of generality, suppose that the elements are
labeled in decreasing order of count, with counts $v_1 \geq v_2 \geq \cdots \geq
v_N$. Write $c = \sum_{i=2}^N v_i^2$, let $k_1$ be the number of hash
functions per trial, and $k_2$ be the number of trials. If we have the condition
\[ v_1 \geq 2\sqrt{\frac{12 k_1 c}{\beta}} +
b(k_1,k_2)\sqrt{6n} \log \left(\frac{6k_1}{\beta}\right) \]
where $b$ is the parameter for $(\epsilon, \delta)$-differential privacy:
\[ b(k_1,k_2) = \frac{\sqrt{8 k_1 k_2 \log(1/\delta)}}{\epsilon} \]
then with probability at least $1-2\beta/3$, the heavy hitter is hashed into the
larger bucket for each hash function in the trial.
\end{lemma}
\begin{proof}
First consider a single hash function. If we define random variables $X_i, i \in
[N]$ by:
\begin{equation*}
  X_i = \left\{
    \begin{array}{ll}
      v_i & : i \mbox{\,is hashed to bucket\,} 1 \\
      0 & : \mbox{otherwise} \\
    \end{array}
  \right.
\end{equation*}
and the function $f(X_2, \cdots, X_N) = \sum_{i = 2}^N X_i$, we show that the
true heavy hitter will be hashed to the larger bucket (with high probability) if
$f$ does not deviate from the mean by too much. If $f$ is close to the mean,
then no matter which bucket the heavy hitter is hashed to, that will become the
larger bucket. However, we will need to keep track of the noise that will be
added to preserve differential privacy. We want $v_1$ to be large enough to
overcome the noise (with high probability).

More precisely, by \Cref{thm-laplace-sum}, the sum of $n$ Laplace
noise terms will be bounded by $b \sqrt{6n}\log(6k_1/\beta)$, with
probability at least $1 - \beta/3k_1$. We also know that the collection $\{X_2,
\cdots, X_N\}$ is a pairwise-independent set of random variables, so applying
\Cref{lemma-k-wise-azuma} with $X = f$, and $t = \sqrt{\frac{12 k_1
c}{\beta}}$, we have that
\[ \Pr[|f - \mu| > t] \leq C_2\left(\frac{2c}{t^2}\right) \leq 4
\left(\frac{c}{t^2}\right) = \frac{\beta}{3k_1} \]
with $C_2$ a constant from \Cref{lemma-k-wise-azuma}.
The difference between the counts in the two buckets will be $2|f - \mu|$, so
for the heavy hitter to be hashed to the larger bucket, we need $v_1 \geq 2 |f -
\mu| + |z|$, where $z$ is the Laplace noise term, with high probability.
Taking a union bound over $k_1$ hash functions, we have that
\[ 2|f - \mu| + |z| \leq 2\sqrt{\frac{12 k_1 c}{\beta}} + b
\sqrt{6n}\log\left(\frac{6k_1}{\beta}\right) \]
holds for all the hash functions in this trial with probability at least $1 -
2\beta/3$. But by assumption, $v_1$ is larger than this gap, and so we are done.
\end{proof}

\begin{lemma}
\label{lemma-unique-wins}
Let the notation be as in the previous Lemma, and consider a single trial in the
Bucket mechanism. If we set $k_1 = \log\left(\frac{3 N}{\beta}\right)$, then with
probability at least $1 - \beta/3$, no other element will be hashed to the same
bucket as the heavy hitter through all the hash functions.
\end{lemma}
\begin{proof}
Pick any element $g$ besides the heavy hitter, and consider a single hash
function.  Since the hash function is pairwise-independent, conditioning on
where the heavy hitter is hashed will not change the marginal for where $g$ will
be hashed. Thus, there is a $1/2$ chance of $g$ colliding with the heavy hitter
for any given hash function. Since the hash functions are drawn independently at
random, the chance of this collision happening on every function is
$(1/2)^{k_1} = \beta/(3N)$, by choice of $k_1$. Taking a union bound over the
$N-1$ elements besides the heavy hitter, we have that this collision probability
for all elements is bounded by $\beta/3$, as desired.
\end{proof}

Now, we are ready to put everything together.

\begin{theorem}
\label{thm-bucket-utility}
Let the notation be as in the previous Lemma. If we set $k_1 = \log(12N), k_2 =
8\log(1/\beta)$, and if we have the condition
\[ v_1 \geq 8 \sqrt{2 c \log(12N)} + \frac{8\log(24 \log(12N))\sqrt{6n\log(12N)
\log (1 / \beta) \log (1 / \delta)}}{\epsilon} =
\]
\[
\tilde{\Omega}\left(\frac{\sqrt{\log |N|}\left(\sqrt{c} + \sqrt{n\log
  \frac{1}{\beta}\log\frac{1}{\delta}}\right)}{\epsilon}\right)
\]
then the Bucket mechanism is $(0, \beta)$-accurate for the heavy hitters
problem.
\end{theorem}
\begin{proof}
First, note that $k_1$ and the condition have been chosen so that from
\Cref{lemma-heavy-wins,lemma-unique-wins}, for any single trial, the heavy
hitter is always hashed to the larger bucket, and is the unique such element,
with probability at least $3/4$. These two conditions ensure that we are able to
correctly identify the heavy hitter with probability $3/4$ for a single trial.
Now, as the trials are independent, we apply a Chernoff bound to show that out
of $k_2$ Bernoulli variables with success probability $3/4$, the probability
that at least half of them succeed is bounded below by \[ \Pr[\mbox{Majority
Vote Success}] \geq 1 - e^{-2 k_2 (1/4)^2} = 1 - \beta \] by our choice of
$k_2$. Thus, the Bucket mechanism returns the true heavy hitter with probability
at least $1 - \beta$.
\end{proof}

We note that the accuracy guarantee of the bucket mechanism is incomparable to
those of our other mechanisms. While the other mechanisms guarantee (without
conditions) to return an element which occurs within some additive factor
$\alpha$ as frequently as the true heavy hitter, the bucket mechanism always
returns the true heavy hitter, so long as a certain condition on $v$ is
satisfied. When the condition is not satisfied, the algorithm comes with no
guarantees. The condition is roughly that the heavy hitter should occur more
frequently than the $\ell_2$-norm of the frequencies of all other elements.
Depending on the distribution over elements, this condition can be satisfied
when the heavy hitter occurs with frequency as small as $\tilde{O}(\sqrt{n})$,
or can require frequency as large as $\Omega(n)$. Finally, we note that this
condition is not unreasonable. It will, for example, be satisfied with high
probability if the frequency of the database elements is follows a power law
distribution, such as a Zipf distribution.

\section{Discussion and Open Questions}
We have initiated the study of the \emph{private heavy hitters} problem in the
fully distributed (local) privacy model. We have provided an (almost) tight
characterization of the accuracy to which the problem can in principle be
solved. In particular, we have separated the local privacy model from the
centralized privacy model: we have shown that even the easier problem of simply
releasing the approximate count of the heavy hitter cannot be accomplished to
accuracy better than $\Omega(\sqrt{n})$ in the local model, whereas this can be
accomplished to $O(1)$ accuracy in the centralized model. We have also given
several efficient algorithms for the heavy hitters problem, but these algorithms
do not in general achieve the tight $\tilde{O}(\sqrt{n\log |N|})$ accuracy bound
that we have established is possible in principle. We leave open the question of
whether there exist \emph{efficient algorithms} in the local model which can
solve the heavy hitters problem up to this information theoretically optimal
bound.

\section*{Acknowledgments}
We would like to thank Raef Bassily and Adam Smith for pointing out an error in
an earlier version of this work.
We would like to thank Martin Strauss for providing valuable clarifications and
insights about~\cite{GSTV07}, and the anonymous reviewers for their helpful suggestions. We
would also like to thank Andreas Haeberlen for suggesting that we study the
heavy hitters problem in the fully distributed setting, and Andreas, Marco Gaboardi, Benjamin
Pierce, and Arjun Narayan for valuable discussions.

\bibliographystyle{alpha}
\bibliography{heavyhitters}
\end{document}